\documentclass[]{spie}  %>>> use for US letter paper
\usepackage{amsmath,amsfonts,amssymb}
\usepackage{graphicx}
\usepackage[colorlinks=true, allcolors=blue]{hyperref}
\usepackage[braket, qm]{qcircuit}
\usepackage{xcolor}
\usepackage{multirow}
\usepackage{booktabs}

\def\k{\tau}
\def\R{\mathbb{R}}

\title{One-step quantum search algorithms based on smooth operators}
\author[a]{Basanta R. Pahari}
\author[b]{Sagar Bhat}
\author[c]{Siri Davidi}
\author[a]{William Oates}

\affil[a]{Department of Mechanical Engineering, Florida Center for Advanced Aero Propulsion, Florida A\&M University and Florida State University, Tallahassee, FL 32310}

\affil[b]{James S. Rickards High School, Tallahassee, FL 32301}
\affil[c]{Lincoln High School, Tallahassee, FL 32311}

\begin{document} 
\maketitle

\begin{abstract}
The discovery of derivatives and integrals was a tremendous leap in scientific knowledge and completely revolutionized many fields, including mathematics, physics, and engineering. The existence of higher-order derivatives means better approximation and, thus, more accurate modeling of any physical phenomenon. Here we use smooth operators that are infinitely differentiable to construct two quantum search algorithms and connect these seemingly different areas. Along with smooth functions,  permutation operators and the roots of unity are exploited to create quantum circuits to perform a quantum search. We validate our models through quantum simulators and test them on IBM's quantum hardware. Furthermore, we investigate the effect of noise and error propagation and demonstrate that our approach is more robust to noise compared to iterative methods like Grover's algorithm.
\end{abstract}
\section{introduction and Motivations}
Orthogonal projection is a useful technique for applications such as pattern recognition and denoising, as it breaks down signals or functions into simpler components, making them easier to analyze~\cite{behrens1994signal,goldstein1998multistage,goldfarb1984unified}. The most common way to create projection operators on functional spaces is by multiplying functions with the characteristic or indicator function associated with the projection domain. However, this method has a significant drawback in that characteristic functions are discontinuous and cannot accurately approximate functions close to the boundary, nor have defined derivatives at the edges. To overcome this, smooth projection operators with higher-order derivatives are desirable. 

An orthogonal projection operator $P$ must be self-adjoint and idempotent, meaning $P^2=P.$ If we take the projection operator $P$ to be $Pf=Gf$ for some function $G$, then
$$P^2f=Pf\iff G^2-G=0\iff G(G-1)=0.$$

Hence G must be 1 or 0 almost everywhere. Any characteristic function satisfies this condition, and therefore, characteristic functions are often used to construct projection operators on functional spaces. To derive a smooth version of this process, some of the authors presented a method using the cyclic permutation operators on $L^2(\mathbb{R}^N)$, the space of square-integrable functions on $\mathbb{R}^N,$\cite{B2019,Pahari2020Dis} based on ideas derived from wavelet theory~\cite{Auscher,H1996}. Here we provide a summary of the process for completion.

To generate orthogonal projections into subspaces of $L^2(\R^N)$ associated with some cone shape partitions of $\R^N$,  let $N \in \mathbb{N}$ and $T$ be the cyclic coordinate permutation operator on $\R^N$. Let $x\in \R^N$,  $s_1\in C^{\infty}(\R^N)$ and  $s_j=s_1(T^{j-1}(\cdot))$ for $j=1,2,...N$ such that $$\sum_{j=1}^{N} |s_j|^2=1.$$

For $f\in L^2(\R^N)$, define $V_j$ as
\begin{eqnarray*}
	V_j f=\frac{1}{\sqrt{N}}\overline{s_j}\sum_{\k=0}^{N-1}w^{\k j}f(T^{\k }(\cdot)),
\end{eqnarray*}
where $w=e^{\frac{2\pi i}{N}}$ is the $N^{th}$ root of unity. Then, the operator $P$ defined as below is an orthogonal projection\cite{B2019,Pahari2020Dis}.

\begin{theorem} \label{thm1}

	\begin{eqnarray*}
		P_j f(x) =V_jV_j^\dagger f(x) =\overline{s_j}\sum_{\k =0}^{N-1}\overline{w^{\k j}}s_j(T^{\k }(x))f(T^{\k }(x))
	\end{eqnarray*}
	is an orthogonal projection on $L^2(\R^N)$. 
	If
	for each $\k \in \{0, 1,2, \dots, N-1\}$, 
	\begin{equation}   \label{eq.sjsi}
	\sum_{j=1}^N w^{\k j} \overline{s_j}(x) {s_{j+\k }}(x) = \delta_{\k,0} \, ,
	\end{equation}
          we also have that
	for any $f \in L^2(\R^N)$
	$$f=\sum_{j=1}^{N}P_j f.$$
	\end{theorem}

The constant functions $s_j= w\frac{1}{\sqrt N}$, for every $j$, and $s_j= \frac{1}{\sqrt N}$, for every $j$, satisfy equation \eqref{eq.sjsi}.{\cite{B2019,Pahari2020Dis}} A non-smooth solution can be obtained by using indicator functions. If  $U_j$ are disjoint subsets of $\R^N$ such that $\chi_{U_{j+1}}=\chi_{U_j}(T(\cdot))$ with $\bigcup_{\k =}^N U_{\k }=R^N$, then $s_j=\chi_{U_{j}}$ is a solution of \eqref{eq.sjsi}. The authors also remark that there is a relationship between the solutions of \eqref{eq.sjsi} and the solution of the filter equations associated with M-band wavelets \cite{LXQH}. 
\begin{figure}[ht]
    \centering
    \includegraphics[scale=0.9]{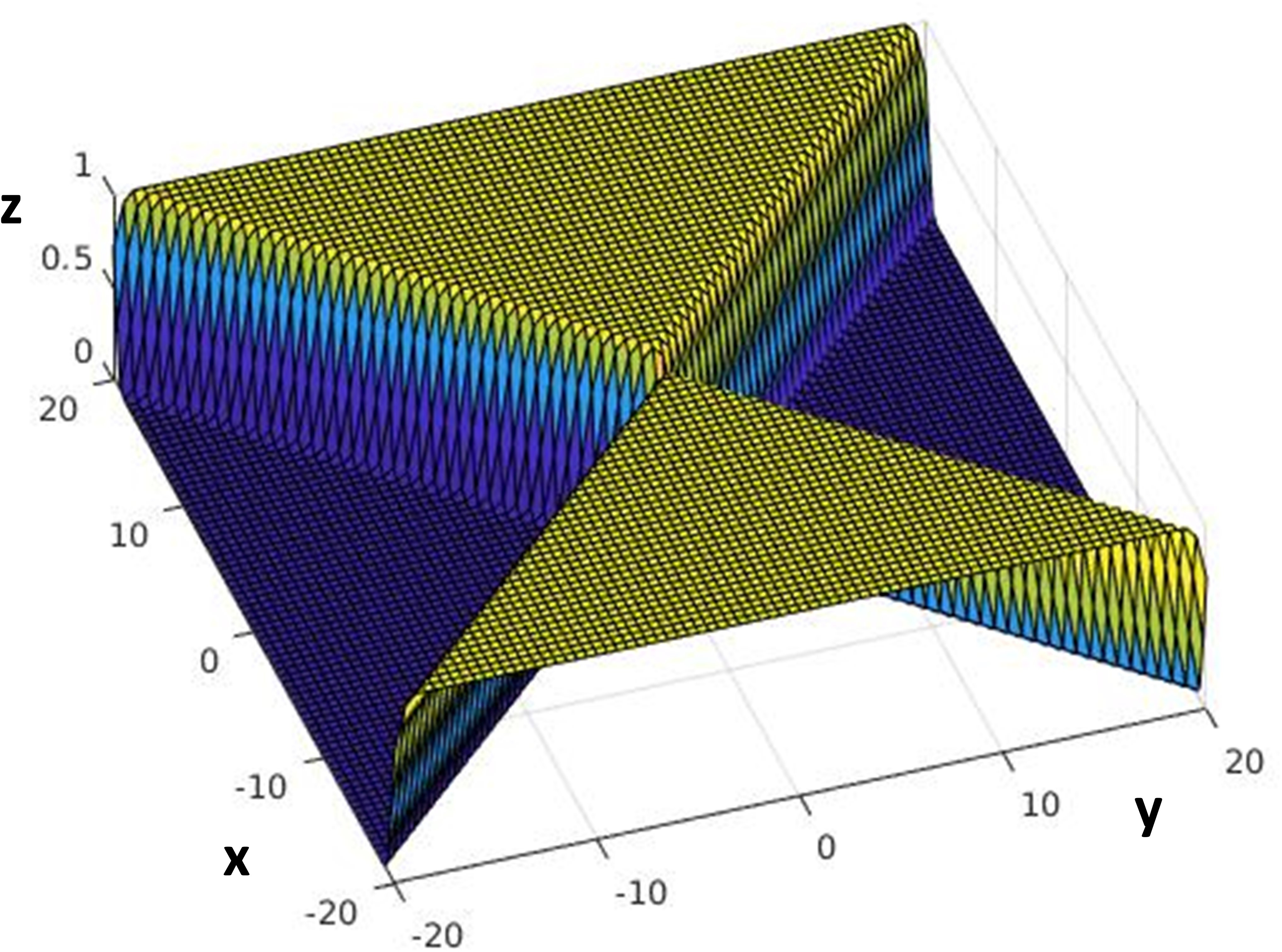}
    \caption{Smooth function $s$ on 2D~\cite{LXQH}. This function is a smooth version of the indicator function of a cone shape domian of $\mathbb{R}^2.$}
    \label{img:smooth}
\end{figure}

In Section~\ref{sec:algo}, we present our quantum algorithm, built upon the process of constructing smooth orthogonal projections described earlier. Remarkably, we find that this process is highly compatible with quantum computing. In Theorem~\ref{thm1}, the projection operator $P$ was created using an explicit symmetrization technique that involves the cyclic permutation operator $T$ and the $N^{th}$ root of unity $w$. Both $T$ and $w$ are cyclic, and the latter lies on the Bloch sphere with a norm of one, which makes them particularly useful for quantum computing. Moreover, by modifying the operators $P$ and $T$, we can transform them into quantum gates that produce unitary operators. Indeed, we have the following result:

\begin{theorem}\label{thm:unitary}
Let $N \in \mathbb{N}$, $T$ be the cyclic coordinate permutation operator on $\R^N,$  $s_1\in C^{\infty}(\R^N)$ and  $s_j=s_1(T^{j-1}(\cdot))$ for $j=1,2,...N$ such that $$\sum_{j=1}^{N} |s_j|^2=1.$$
Assume for each $\k \in \{0, 1,2, \dots, N-1\}$, 
	$$   
	\sum_{j=1}^N w^{\k j} \overline{s_j}(x) {s_{j+\k }}(x) = \delta_{\k,0} \,$$ 
 where $\delta_{\k,0}$ is the Kronecker delta. Then for $j,k=1,2,..,N$ the matrix $U$ defined as $$U_{jk}=\overline{s}_jw^{jk},$$
  is  unitary.  
\end{theorem}
\begin{proof}
We have that

\begin{align*}
(U^\dagger U)_{jk} = \sum_{l=1}^N \overline{U}_{lj} U_{lk} = \sum_{l=1}^N (\overline{\overline{s}_lw^{lj}})( \overline{s}_lw^{lk})= \sum_{l=1}^N ({s}_lw^{-lj})( \overline{s}_lw^{lk})
= \sum_{l=1}^N \overline{s_l}s_lw^{(k-j)l} 
.
\end{align*}
Since $\sum_{j=1}^{N} |s_j|^2=1$, we have $\sum_{j=1}^{N} s_j\overline{s}_j=1$, so
\begin{align*}
(U^\dagger U)_{jk} =\sum_{l=1}^N \overline{s_l}s_lw^{(k-j)l} = \delta_{k-j,0}=\delta_{j,k}.
\end{align*}Thus $U^\dagger U = I$ and similarly one can show $UU^\dagger= I$ and therefore $U$ is unitary.
\end{proof}

Note that the double index notation used in this context differs from Einstein's summation. In particular, the notation $U_{jk}$ denotes the element located in the $j^{th}$ row and $k^{th}$ column of the matrix $U$. Additionally, the expression $w^{jk}$ represents $w$ raised to the power of $j$ multiplied by $k$. Theorem~\ref{thm:unitary} imposes certain restrictions on the smooth function $s_j$, but there exist numerous smooth functions that satisfy it. Just as with Theorem~\ref{thm1}, the constant functions $s_j= w\frac{1}{\sqrt N}$ for every $j$, and $s_j= \frac{1}{\sqrt N}$ for every $j$ are both solutions to Theorem~\ref{thm:unitary}. In our quantum circuit, we will use $s_j= \frac{1}{\sqrt N}$, but depending on the task at hand, one could construct a more complex function $s_j$, as illustrated in Figure~\ref{img:smooth}. For instance, multiple sparse representation systems of Parseval frames were smoothly combined using the function depicted in Figure~\ref{img:smooth}, and additional research on related topics has also been pursued. \cite{EDB2021,bownik2015}.

Since each quantum gate for a single qubit is represented by a $2 \times 2$ matrix, we make a special remark on 2D unitary operators created by the above process. In 2D, the requirement that for each $\k \in \{0, 1,2, \dots, N-1\}$, 
	$$   
	\sum_{j=1}^N w^{\k j} \overline{s_j}(x) {s_{j+\k }}(x) = \delta_{\k,0} \,$$ 
 is satisfied by all 2D real functions $s_j$ with the property  $$\sum_{j=1}^{N} |s_j|^2=1.$$ 
This follows because the second root of unity is $w=e^{i \pi}=-1.$ In fact, for a real valued 2D function $s_1$, the matrix $U$  has the following form

 \[ U=\begin{bmatrix} 
    s_1w^0  &   s_2w^1  \\
   s_2w^0 &  s_1w^2  \end{bmatrix}  =
    \begin{bmatrix}s_1  &   -s_2  \\
   s_2 &  s_1  
    \end{bmatrix}.\] Then clearly

     \[ UU^\dagger=\begin{bmatrix} 
    s_1  &   -s_2 \\
   s_2 &  s_1 \end{bmatrix} 
    \begin{bmatrix}
    s_1  &   s_2  \\
   -s_2 &  s_1  
    \end{bmatrix}=\begin{bmatrix}
    {s_1}^2+{s_2}^2  &   s_1s_2-s_2s_1  \\
   s_2s_1-s_1s_2 &  {s_1}^2+{s_2}^2  
    \end{bmatrix}=\begin{bmatrix}
    1  &   0  \\
   0 &  1  
    \end{bmatrix}.\]

The cyclic permutation operator $T$ is the second tool we use to build quantum gates. To illustrate, consider a sequence of numbers where all elements are zero, except for one with a value of 1. By combining all possible cyclic permutations of this sequence, we can create a unitary operator. We will demonstrate this process of building unitary operators in the next section. The Quantum Cyclic Permutation Algorithm (QCPA), which we developed in Section~\ref{first:algo}, shares similarities with the Bernstein-Vazirani algorithm~\cite{Bernstein1997}. However, our approach was inspired by function analysis and sparse representation theory.

\section{Quantum Search Algorithm}\label{sec:algo}
Grover's search algorithm is a crucial algorithm in the field of quantum computing that allows for searching through an unsorted database at a much faster rate than classical algorithms~\cite{grover1996fast}. This algorithm provides a quadratic speedup over classical search methods, making it a vital tool for solving optimization and search problems across various fields, including cryptography, machine learning, and artificial intelligence. The algorithm is composed of three significant steps, namely superposition, phase flip, and amplitude amplification, achieved via the application of the Hadamard gate, Grover oracle, and diffusion operator, respectively, as shown in Figure~\ref{grover}. The inversion about the mean, the last step of the algorithm, is particularly useful for amplifying the probability amplitude of the marked state and has been applied to solve problems in other fields, including linear algebra~\cite{brassard2002,ambainis2012}.

\begin{figure}[h]
    \centering
    \includegraphics[scale=0.45]{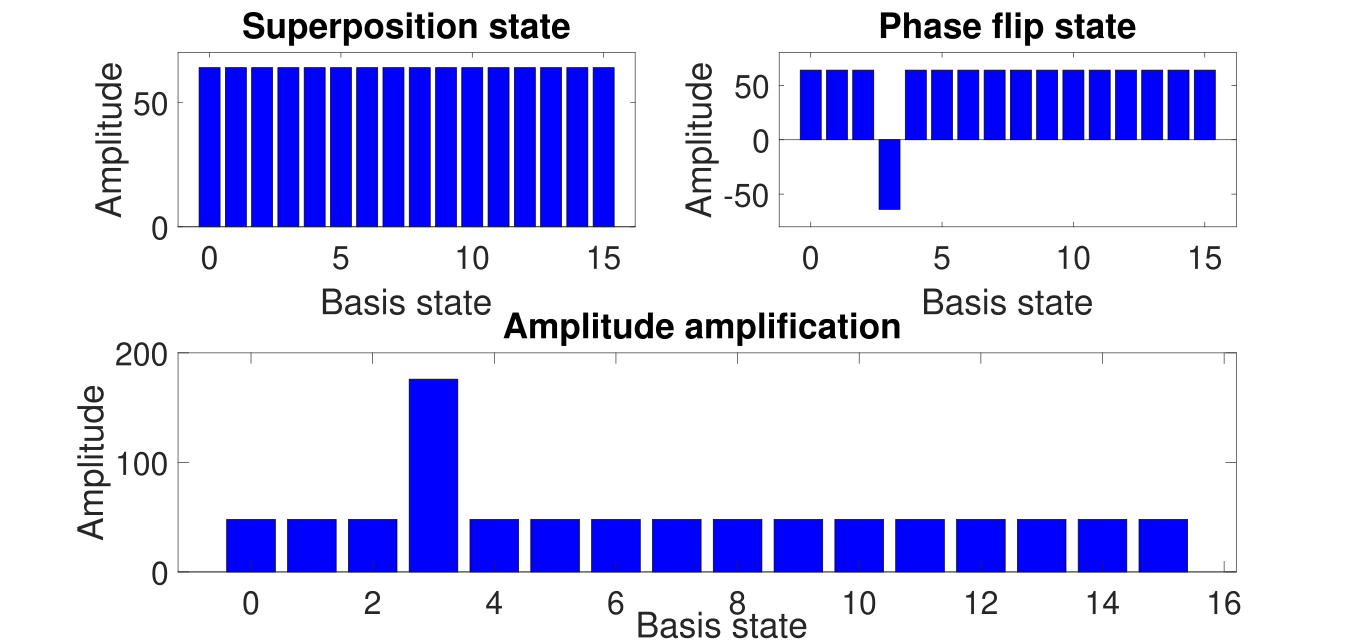}
    \caption{Three steps of Grover's Search algorithm: superposition, phase flip and amplitude amplification.}
    \label{grover}
\end{figure}  

This paper's primary purpose is not to directly compare the runtime complexity between our proposed algorithm and Grover's search algorithm. Instead, we aim to investigate the potential of constructing algorithms that outperform theoretically optimal algorithms under certain conditions. Quantum systems are subject to inherent problems such as noise, coherence, and measurement accuracy, which can affect the performance of quantum algorithms in practice~\cite{preskill2018quantum}. Furthermore, exploring the relationship between functions with higher-order derivatives may provide insights into solving quantum approximation problems, for instance, state estimations~\cite{lloyd2013quantum}.

\subsection{Quantum Cyclic Permutation Algorithm (QCPA)}\label{first:algo}
We now present a new quantum search algorithm  based on the cyclic-permutation oracle. 
 Suppose we have a sequence $x=(x_i)$ for $i=1,2,...,N$ such that
 
 \[ x_i=\begin{cases} 
      1 & i= s \\
      0 & \text{otherwise}.
   \end{cases}
\]
We construct a quantum circuit capable of finding the ``marked state" $s.$

Let $w=\exp(\frac{2\pi i}{N})$ be the $N^{th}$ root of unity and 

\[
U = \frac{1}{\sqrt{N}}\begin{bmatrix} 
    {w}^{0\cdot 1} & w^{1\cdot 1}  & \dots &  w^{(N-1)\cdot 1}& \\
    {w}^{0\cdot 2} & w^{ 1 \cdot 2}  & \dots &  w^{(N-1)\cdot 2}&  \\
    \vdots & \vdots  & \ddots \   &\vdots &\\
    {w}^{0\cdot (N-1)} & w^{1\cdot (N-1)}  & \dots &  w^{(N-1)\cdot (N-1)}& \\
    {w}^{0\cdot N} & w^{ 1 \cdot N}  & \dots &  w^{(N-1)\cdot N}& 
    \end{bmatrix}=\frac{1}{\sqrt{N}} \begin{bmatrix} 
    {w}^{0} & w^{1}  & \dots &  w^{N-1}& \\
    {w}^{0} & w^{ 2}  & \dots &  w^{N-2}&  \\
    \vdots & \vdots  & \ddots \   &\vdots &\\
    {w}^{0} & w^{N-1}  & \dots &  w^{1}& \\
    {w}^{0} & w^{0}  & \dots &  w^{0}& 
    \end{bmatrix}\]

The $N\times N$ unitary matrix $U$  is the $N$ dimensional version of the matrix constructed in Theorem~\ref{thm:unitary} with $s_j=\frac{1}{\sqrt{N}}$ and hence can act as a quantum gate. Additionally, for any scalar vector $v=[v_1\quad v_1\quad \cdot \cdots \quad v_1]^T$, it possesses the property: 

 \[ Uv=\begin{cases} 
     \sqrt{N}v_1 & i= N \\
      0 & \text{otherwise}.
   \end{cases}
\]

Furthermore, let us suppose that we have an oracle $F$ capable of shifting all quantum states by $s$ unit counter-clockwise. This oracle can be constructed using a cyclic permutation of the sequence $x$, where

\[F = \begin{bmatrix} 
    x_1 & x_2  \qquad & \dots &  x_N& \\
    x_2 & x_3\qquad   & \dots &  x_1   \\
    \vdots \qquad  & \vdots  & \ddots \   &\vdots &\\
    x_{N-1} & x_{N} \qquad   & \dots &  x_{N-2}&\\
    x_{N}  & x_1  & \qquad  \dots &  x_{N-1}& 
    \end{bmatrix}.\]

Given that the sequence $x$ contains only one non-zero element, which is $1$, the columns of the matrix $F$ are orthonormal, making $F$ unitary and therefore a quantum gate. The oracle $F$ is a permutation matrix and shares similarity with the oracle used in Grover's search. Grover's oracle flips the phase of the $s^{th}$ quantum state while $F$ cyclically shifts all states $s$ times. With the necessary setup complete, we can now outline our quantum algorithm, which involves three primary steps:

\begin{itemize}
\item \textbf{Superposition:} achieved by applying the Hadamard gate
\item \textbf{Collapse:} collapsing the probabilities to the last quantum state using the $U$ gate
\item \textbf{Shift:} shifting all quantum states counter-clockwise by $s$ units using the $F$ gate.
\end{itemize}

\begin{figure}
%\centering
\hspace{25mm}
%\begin{center}
\Qcircuit @C=4em @R=3em {
   &\lstick{\ket{0}} &\gate{H} &\multigate{3}{U} &\qw &\multigate{3}{F} &\qw  &\meter \\
   &\lstick{\ket{0}} &\gate{H} &\ghost{U} &\qw &\ghost{F} &\qw  &\meter \\
   &\lstick{\ket{0}} &\gate{H} &\ghost{U} &\qw &\ghost{F} &\qw  &\meter \\
   &\lstick{\ket{0}} &\gate{H} &\ghost{U} &\qw &\ghost{F} &\qw  &\meter \\
}\\
%\end{center}
\caption{Circuit representation for QCPA}
\label{img:circuit}
\end{figure}
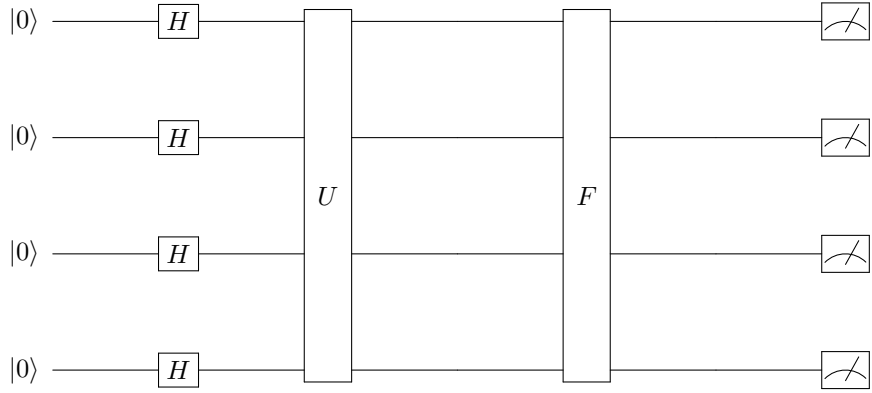

Finally, we can determine the location of the marked state $x$ by calculating $s=N+1-j$, where $j$ is the quantum state obtained after measurement. This is because the cyclic permutation operator shifts the location of quantum states by $s$ units.

In the following discussion, we will demonstrate the workings of our algorithm for a three-qubit system, where the state $\ket{101}$ (sixth location of the sequence)  is marked.

\begin{enumerate}
\item Apply Hadamard gates to $3$  qubits initialized to  $ \ket{000}$  to create a uniform superposition:
$$\ket{{\Psi}_{H}}=\frac{1}{\sqrt{8}}\big(\ket{000}+\ket{001}+\ket{010}+\ket{011}+\ket{100}+\ket{101}+\ket{110}+\ket{111}\big)$$

\item Collapse the probabilities of all quantum states except the last state by applying gate $U$ to obtain one state with a probability of one:
$$\ket{{\Psi}_{U}}=U\ket{{\Psi}_{H}}=\frac{1}{\sqrt{8}}\cdot \frac{8}{\sqrt{8}}\ket{111}=\ket{111}$$

\item Shift the remaining state $s$ unit counter-clockwise using oracle $F$, where $s$ is the location of the marked state:
$$\ket{{\Psi}_{F}}=F\ket{{\Psi}_{U}}=F\ket{111}=\ket{010}$$

\end{enumerate}

As a result of the measurement, the quantum state $\ket{010}$ is be obtained, corresponding to the third position in the sequence. Therefore, the $1$ in the initial sequence is located at position $s=N+1-3=8+1-3=6,$ which is the state $\ket{101}$.

\subsubsection{Understanding QCPA via matrix operations}

Suppose we have a sequence $x=0, 1, 0, 0$. Since $N=4$,  $w=\exp(\frac{2\pi i}{4})=\exp(\frac{\pi i}{2})=i.$
Then 
\[ U =\frac{1}{\sqrt{4}} \begin{bmatrix} 
    {i}^{0\cdot 1} & i^{1\cdot 1}  &   i^{2\cdot 1}&  i^{3\cdot 1}&  \\
    {i}^{0\cdot 2} & i^{1\cdot 2}  &   i^{2\cdot 2}&  w^{3\cdot 2}&  \\
    {i}^{0\cdot 3} & i^{1\cdot 3}  &   i^{2\cdot 3}&  w^{3\cdot 3}& \\
    {i}^{0\cdot 4} & i^{1\cdot 4}  &   i^{2\cdot 4}&  w^{3\cdot 4}& 
    \end{bmatrix} = \frac{1}{2} \begin{bmatrix} 
    {i}^{0} & i^{1}  &   i^{2}&  i^{3}& \\
   {i}^{0} & i^{2}  &   i^{0}&  i^{2}&  \\
    {i}^{0} & i^{3}  &   i^{2}&  i^{1}& \\
    {i}^{0} & i^{0}  &   i^{0}&  i^{0}& 
    \end{bmatrix}=\frac{1}{2} \begin{bmatrix} 
    1 & i  &  -1&  -i& \\
   1 & -1  &   1&  -1&  \\
    1 & -i  &   -1&  i& \\
    1 & 1  &   1&  1& 
    \end{bmatrix}\]
    
    and
    \[ F=\begin{bmatrix} 
    0  &   1&  0& 0&\\
   1 & 0  &   0&  0&  \\
    0 & 0 &   0&  1& \\
    0 & 0  &   1&  0& 
    \end{bmatrix}\]

     \begin{figure}[h]
    \centering
    \includegraphics[scale=0.40]{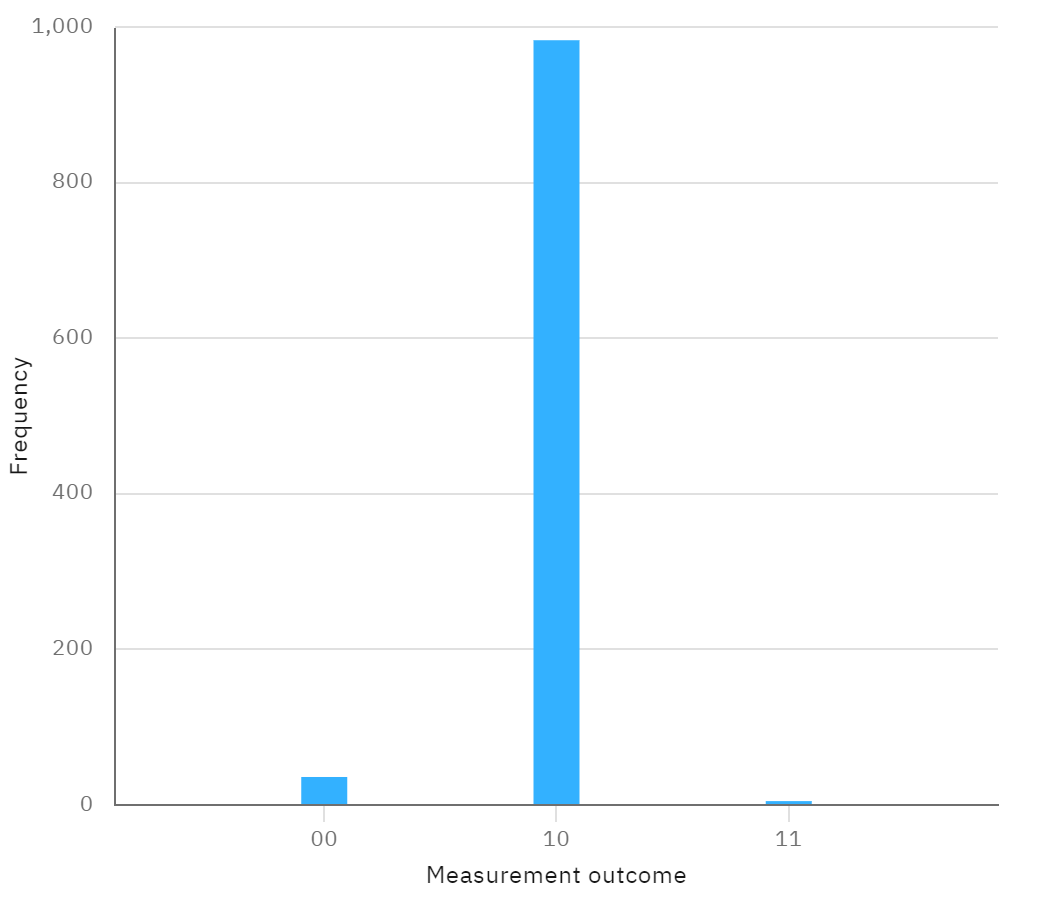}
    \caption{Histogram of measurement outcomes on IBM hardware for finding the location 1 in a sequence of four numbers. The state $\ket{10}$ is the expected state predicted by QCPA.}
    \label{2qubidhardware}
\end{figure}  
    
Now, if we initialize 2 qubits to $0$ states and feed them through the Hadamard gate, we get a superposition state  $q=[ q_1 \quad q_2 \quad q_3 \quad q_4]^T,$ where $|q_i|=\frac{1}{2}.$
For simplicity, we write $q=[ 0.5 \quad 0.5 \quad 0.5 \quad 0.5]^T.$

\[FUq=\frac{1}{2} \begin{bmatrix} 
   0  &   1&  0& 0&\\
   1 & 0  &   0&  0&  \\
    0 & 0 &   0&  1& \\
    0 & 0  &   1&  0& 
    \end{bmatrix}\begin{bmatrix} 
    1 & i  &  -1&  -i& \\
   1 & -1  &   1&  -1&  \\
    1 & -i  &   -1&  i& \\
    1 & 1  &   1&  1& 
    \end{bmatrix}
    \begin{bmatrix} 
    0.5 \\
   0.5  \\
    0.5  \\
    0.5  
    \end{bmatrix}\]
    
    \[\qquad =\frac{1}{2}\begin{bmatrix} 
  0  &   1&  0& 0&\\
   1 & 0  &   0&  0&  \\
    0 & 0 &   0&  1& \\
    0 & 0  &   1&  0& 
    \end{bmatrix} \begin{bmatrix} 
    0 \\
   0  \\
    0  \\
    2  
    \end{bmatrix}\]

   \[= \begin{bmatrix} 
    0 \\
   0  \\
    1  \\
    0  
    \end{bmatrix}\]

This implies that the third quantum state will appear with a probability of $1$ if a measurement is performed. Figure~\ref{2qubidhardware} presents an experimental verification of the algorithm's accuracy in determining the third quantum state (i.e., \ket{10}) using IBM hardware. As anticipated, the results demonstrate that the algorithm accurately predicted the correct quantum state. Once a measurement is performed, the marked location is given by  $s=N+1-j=5-3=2.$

The overall complexity of constructing the permutation quantum gate using a sequence of the length of $N$ is $O(N^2).$ This is because the process of creating each row of the gate requires $O(N)$ operations, and since there are $N$ rows, the total complexity is $O(N^2).$ Once the gate is built, we only require one step operation for detecting the marked state, so the complexity of the entire algorithm is $O(N^2).$ Despite Grover's method having lower computational complexity than QCPA and QUSA, which is presented in Section~\ref{sec:second}, several advantages are associated with these algorithms. For instance, both QCPA and QUSA demonstrate lower noise susceptibility, as seen in Section~\ref{noise:sec}. Additionally, these methods contribute to advancing our theoretical comprehension of search algorithms and offer a pathway for constructing quantum circuits utilizing smooth operators.

\subsection{Quantum Unity Sum Algorithm (QUSA)}\label{sec:second}
The Quantum Unity Sum Algorithm (QUSA) is the second novel search algorithm presented in this work. It bears similarities to the QCPA introduced in Section~\ref{first:algo}, but it also features two distinctive elements. Firstly, QUSA employs a function $f$, similar to the one used in Grover's search, which outputs $0$ for the marked state and $1$ for all other states. Secondly, instead of explicitly using the permutation oracle, QUSA adjusts the unitary operator $U$ by leveraging a well-known property of the $N^{th}$ root of unity. This property allows the probability to be shifted to the correct quantum state. If
\[
U = \frac{1}{\sqrt{N}} \begin{bmatrix} 
    {w}^{0} & w^{1}  & \dots &  w^{N-1}& \\
    {w}^{0} & w^{ 2}  & \dots &  w^{N-2}&  \\
    \vdots & \vdots  & \ddots \   &\vdots &\\
    {w}^{0} & w^{N-1}  & \dots &  w^{1}& \\
    {w}^{0} & w^{0}  & \dots &  w^{0}& 
    \end{bmatrix}\]

then each element of the matrix $U$ in the $k$-th row and $l$-th column can be represented as 

$$U_{kl} =\frac{1}{\sqrt{N}}w^{k(l-1)}$$

\iffalse

\begin{figure}[h]
    \centering
    \includegraphics[scale=0.6]{Fig/2ndalgoloc2.png}
    \caption{Quantum circuit for detecting the location of $1$ described in Section~\ref{sec:second} for the sequence $x=0,1,0,0$. Here, $U_3$ is a universal basis, and the circuit is created by transpiling our quantum circuit using the Python Qiskit module.   }
    \label{qc2qubitloc2}
\end{figure}  
\fi
\begin{figure}[ht]
\hspace{35mm}
\Qcircuit @C=4em @R=3em {
   &\lstick{\ket{0}} &\gate{H} &\multigate{3}{\tilde{U}}    &\qw  &\meter \\
   &\lstick{\ket{0}} &\gate{H} &\ghost{U} &\qw    &\meter \\
   &\lstick{\ket{0}} &\gate{H} &\ghost{U} &\qw   &\meter \\
   &\lstick{\ket{0}} &\gate{H} &\ghost{U} &\qw    &\meter \\
}\\

  \caption{Circuit representation for QUSA}
    \label{sec:circ }
\end{figure}
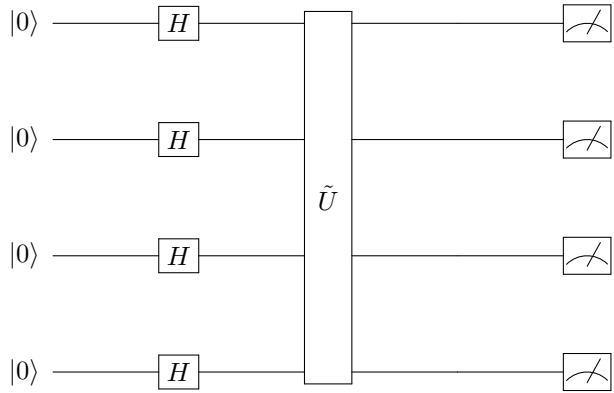
We construct a new oracle $\Tilde{U}$ capable of locating the marked state. For $k=1,2,....,N-1$ let

$$\Tilde{U}_{kl} =({U}_{kl})^{(f(x_k))}$$

and $S=\frac{1}{\sqrt{N}}\bigg(-\sum_{k=1}^{N-1}{U}_{kN}\bigg)^{-1}.$ The last row of $\Tilde{U}$ is given by 

$$\Tilde{U}_{Nl} = \frac{1}{\sqrt{N}}{S^{l-1}}.$$

 The new unitary matrix $\Tilde{U}$ has the property that for any scalar vector $v=[v_1\quad  v_1\quad  \cdot \cdots \quad v_1]^T,$

 \[ \Tilde{U}v=\begin{cases} 
     \sqrt{N}v_1 & i= s \\
      0 & \text{otherwise}.
   \end{cases}
\]
where $s$ is the marked state.

This approach involves a quantum circuit with the following sequence of steps:
\begin{itemize}
    \item Set all quantum states to the zero state
    \item Use the Hadamard gate to create superposition
    \item Apply the gate $\Tilde{U}$
    \item Perform a measurement 
   
\end{itemize}

This completes the construction of our new circuit, and Figure~\ref{algo23qubitloc4} displays the measurement results obtained by running this circuit on an IBM quantum computer to locate the marked state of a sequence. This algorithm exploits a key property of the $N^{th}$ root of unity, namely, that the sum of the roots of unity is zero. The procedure involves generating $N-1$ rows of the unitary operator $\Tilde{U}$ by exponentiating the $i^{th}$ row of $U$ by $f(x_i)$. Since $f(x_i)$ equals 1 for all but the marked state $s$, the rows of $\Tilde{U}$ corresponding to $x_i$ other than $s$ remain identical to the corresponding rows of $U$, while the row corresponding to $s$ is set to all ones. To ensure that $\Tilde{U}$ is isometric, we construct its last row using the $N-1$ previously constructed rows.

\begin{figure}[ht]
    \centering
    \includegraphics[scale=0.35]{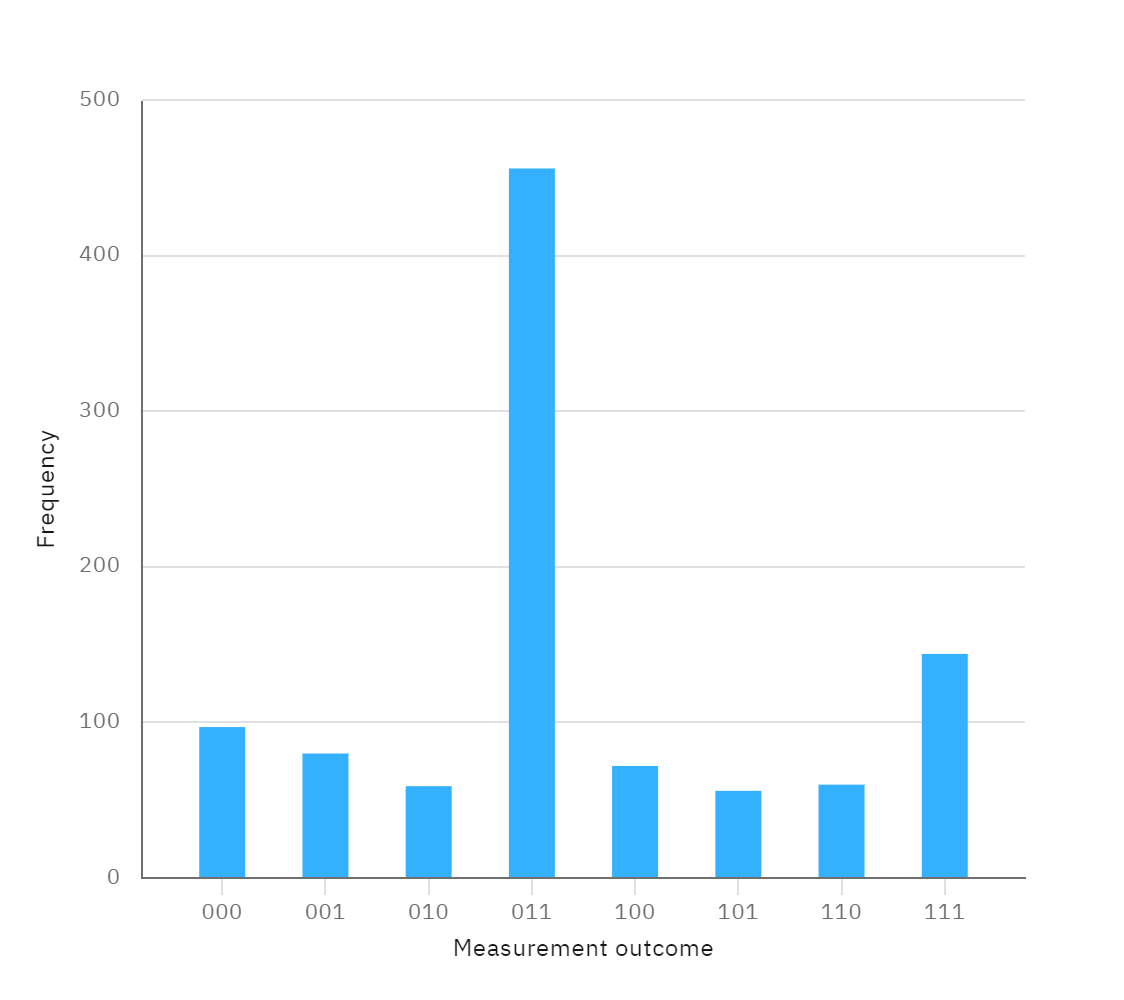}
    \caption{Histogram of measurement outcomes on IBM hardware for finding the location 1 in a sequence of eight numbers. The state $011$ is the expected state predicted by QUSA described in Section~\ref{sec:second}.}
    \label{algo23qubitloc4}
\end{figure}  
\section{Noise Analysis}\label{noise:sec}
Quantum error correction and related topics have been extensively studied in quantum computing, given the susceptibility of quantum hardware and measurement instruments to errors ~\cite{devitt2013quantum,lidar2013quantum}. Error propagation can cause poor performance in algorithms requiring more gate operations or larger circuits. Here, we investigate how our algorithms perform when noise is added to the system. We demonstrate that our one-step search algorithm is less susceptible to noise than Grover's search algorithm, which requires multiple iterations. For this purpose, we performed two sets of experiments to verify this claim: noise-added simulations and actual quantum hardware testing. All the quantum circuits we used in this study were built and transpiled using the Qiskit package in Python. We used the universal gates CX and U3 in the circuit construction process. To simulate a noisy system, we added depolarizing errors to single-qubit U3 and CX gates, which can model the effect of imperfect control or environmental interaction. The depolarizing error was applied with the probability of 0.001 gate error for both types of gates. We applied an equal number of noisy gates to both our algorithms and Grover's search algorithm to examine the impact of error propagation on their performance and to ensure a fair comparison of the three methods. We applied noise to the first n-gates and let the noise propagate. Since the Grover search algorithm is iterative, it necessitates more gate operations than our approach as the number of states grows. Figure~\ref{gatesbarplot} presents a bar graph depicting the number of gates required to execute each algorithm via the transpiler. As the sequence order we are seeking rises, the Grover technique demands considerably more gate operations than the other two. As expected, Grover's search algorithm predicted the right quantum state but with much less accuracy than our algorithm. Table~\ref{tab:quantum_perf} presents the simulation results, while Figure~\ref{gatesbarplot} shows that the QUSA requires fewer gates than QCPA to run, and consequently, the QUSA is less susceptible to noise than QCPA.

\begin{table}[htbp]
  \centering
  \caption{Comparison of Accuracy }
  \label{tab:quantum_perf}
  \begin{tabular}{|c|ccc|ccc|ccc|}\hline {Qubits}
%    \hline
    & \multicolumn{3}{c|}{Grover} & \multicolumn{3}{c|}{QCPA} & \multicolumn{3}{c|}{QUSA} \\
    \cline{2-10}\rule{0pt}{4ex}  
    & \multicolumn{3}{c|}{ Number of Noisy Gates} & \multicolumn{3}{c|}{Number of Noisy Gates} & \multicolumn{3}{c|}{Number of Noisy Gates} \\
        & 2 & 4 & 6 & 2 & 4 & 6 & 2 & 4 & 6  \\
    \hline \rule{0pt}{4ex}  

    3 & 90.62 & 90.14 & 88.09 & 98.54 &  97.85 &  96.97 & 98.83 & 97.95 & 97.46  \\
  %  \cline{1-13}
 
      4    & 76.66 & 61.72 & 58.50 & 89.84 & 82.42 & 81.35 & 94.82 & 90.04 & 88.96  \\
  %  \cline{1-13}
    
       5   & 30.47 & 11.33 & 6.93 & 66.41 & 42.97 & 40.82 &  78.52 & 66.80 & 58.40\\

    \hline
    
    \hline
  \end{tabular}
\end{table}

\begin{figure}[ht]
    \centering
    \includegraphics[scale=0.45]{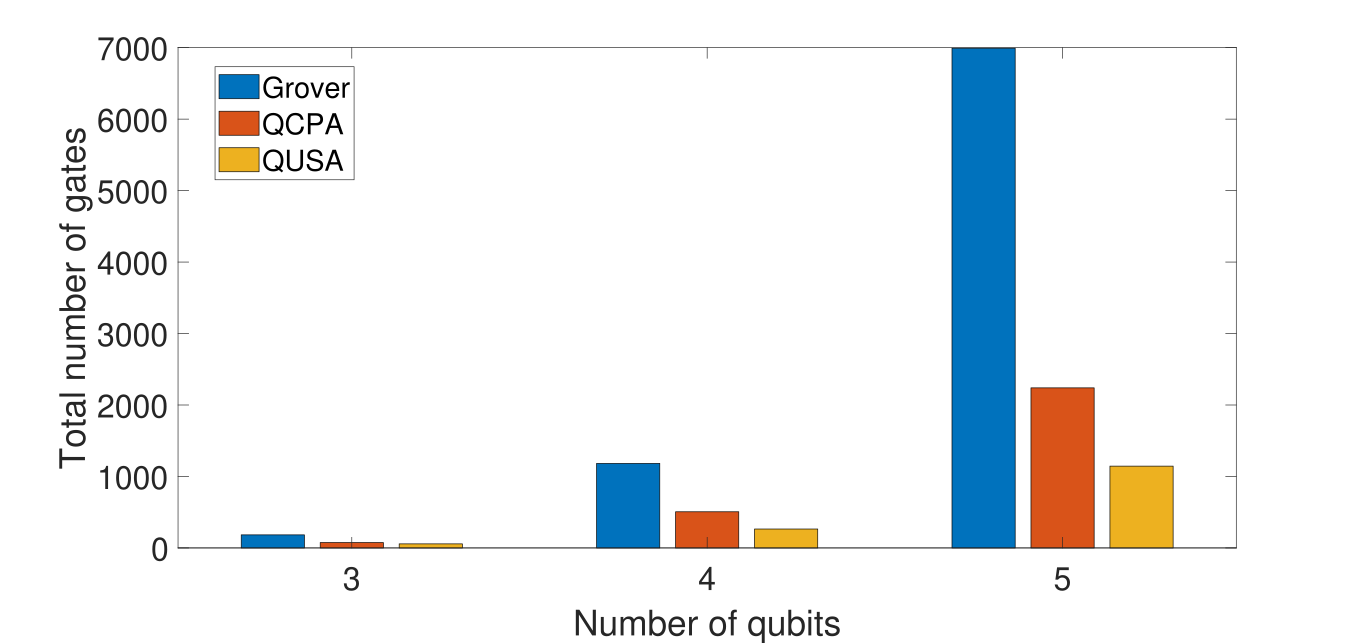}
    \caption{Number of gates in transpiled circuits for each of the three methods.}
    \label{gatesbarplot}
\end{figure}  

To further validate our assumption, we tested both algorithms on actual quantum hardware, which naturally introduces errors due to measurement and other quantum system issues. Our algorithm outperformed Grover's search, just as in the noisy simulation case. The measurement results for a two-qubit system consisting of four states, to predict the marked state of 10 using both algorithms and Grover's method, are presented in Table~\ref{Table:Errors}. Grover's method had a success rate of $94.04\%$, while our algorithm performed slightly better with a success rate of $96.19\%$. The QUSA also performed well, with a success rate of $95.60\%$. These findings support our assumption that one-step search is less error-prone compared to Grover's search, which can exacerbate errors at each iteration. This is particularly relevant since quantum systems are highly susceptible to measurement and other errors.

\begin{table}[htbp]
\caption{Measurement results using Grover's method, QCPA and QUSA for a two-qubit system (4 states) in predicting the marked state $\ket{10}$ on IBM quantum computer (ibm\_jakarta). The simulation was conducted $1024$ times for each method.  }
\centering
% \usepackage{multirow}

{\begin{tabular}{c|cccc|c} \hline \hline{Algorithm  }
\multirow{3}{*} & \multicolumn{4}{c|}{Quantum State} & \multicolumn{1}{c}{Performance(\%)}   \\
                       & 00 & 01 & 10 & 11 &  \\ \hline \rule{0pt}{4ex}  
Grover             &     50     &   3    &  963  &     8  & 94.04      \\ 
QCPA                &     27     &    0   &   985  &      12 &  96.19    \\ 
QUSA                &     25    &    10   &   979  &      10 &  95.60     \\\hline \hline 
\end{tabular}}

\label{Table:Errors}
\end{table}

\section{Conclusion}\label{sec:conclusion}

Previous related works in shearlet and wavelets  used permutation operators to achieve symmetrization, ensuring that the functions remained infinitely differentiable at the intersection of N-dimensional partitions~\cite{Pahari2020Dis,B2019,Auscher}. However, we used permutation operators for a different purpose: to shift quantum states. Furthermore, smooth projection operators and shearlet theory utilized the sum of the roots of unity equals zero to achieve perfect signal reconstruction by carefully balancing the contributions from each partition of the projected
signal~\cite{EDB2021}. In contrast, our approach involves utilizing this same property of the roots of unity to construct a unitary operator
that shifts the probabilities toward the marked quantum state. This demonstrates that concepts derived from wavelet and shearlet theories can be adapted to address quantum computing problems. Additionally, our algorithm demonstrated better accuracy than the Grover search algorithm in predicting the correct quantum state. This enhanced performance may be attributed to the fact that our one-step search algorithms do not rely on the Grover oracle explicitly and are non-iterative, unlike Grover's algorithm.

We introduced two novel search algorithms that rely on smooth operators. Our primary objective was to establish a link between infinitely differentiable operators and quantum computing, which could lead to the development of new quantum algorithms capable of leveraging the excellent approximation properties of derivatives. Although we employed a single smooth function $ s_j=\frac{1}{N}$ to construct both algorithms, other smooth functions could also be utilized to construct quantum operators that address the given problems. Additional research could explore ways to decrease the runtime complexity of both methods and develop novel quantum algorithms that utilize smooth functions. These investigations need not be limited to search algorithms but can encompass a broader range of applications within quantum computing.  This approach has the potential to expand our theoretical understanding of quantum computing significantly.


%\input{6_Appendix}
%\newpage
% References
\bibliography{main} % bibliography data in report.bib
\bibliographystyle{plain} % makes bibtex use spiebib.bst

\end{document}